\documentclass{amsart}
\usepackage{algorithm,algpseudocode}
\newtheorem{theorem}{Theorem}[section]
\newtheorem{lemma}[theorem]{Lemma}

\theoremstyle{definition}
\newtheorem{definition}[theorem]{Definition}

\theoremstyle{remark}
\newtheorem{remark}[theorem]{Remark}
\newtheorem{corollary}[theorem]{Corollary}
\numberwithin{equation}{section}
\usepackage{tabularx}
\usepackage{booktabs}

\begin{document}
\title[Exact Digit-by-Digit Root Algorithms]
{Exact Constructive Digit-by-Digit Algorithms for Integer $e$-th Root Extraction}

\author{Suresan Pareth}
\address{Department of Information Technology, National Institute of Technology Karnataka, Surathkal Mangalore, India}
\email{al.sureshpareth.it@nitk.edu.in, sureshpareth@gmail.com}

\date{\today}


\begin{abstract}
   \begin{abstract}
       We present a unified constructive digit-by-digit framework for exact root
       extraction using only integer arithmetic. The core contribution is a
       complete correctness theory for the fractional square root algorithm,
       proving that each computed decimal digit is exact and final, together with
       a sharp truncation error bound of $10^{-k}$ after $k$ digits.

       We further develop an invariant-based framework for computing the integer
       $e$-th root $\lfloor N^{1/e} \rfloor$ of a non-negative integer $N$ for
       arbitrary fixed exponents $e \ge 2$, derived directly from the binomial
       theorem. This method generalizes the classical long-division square root
       algorithm, preserves a constructive remainder invariant throughout the
       computation, and provides an exact decision procedure for perfect $e$-th
       power detection.

       We also explain why exact digit-by-digit fractional extraction with
       non-revisable digits is structurally possible only for square roots
       ($e=2$), whereas higher-order roots ($e \ge 3$) exhibit nonlinear coupling
       that prevents digit stability under scaling.

       All proofs are carried out in a constructive, algorithmic manner consistent
       with Bishop-style constructive mathematics, yielding explicit algorithmic
       witnesses, decidable predicates, and guaranteed termination. The resulting
       algorithms require no division or floating-point operations and are well
       suited to symbolic computation, verified exact arithmetic, educational
       exposition, and digital hardware implementation.
   \end{abstract}

\end{abstract}

\maketitle



\section{Introduction}
\label{sec:introduction}

The digit-by-digit computation of integer roots is a classical technique rooted in elementary arithmetic. The familiar long-division algorithm for square roots, which processes pairs of digits and incrementally builds the root while preserving a non-negative remainder, admits a natural algebraic generalisation to higher-order roots via the binomial theorem. This paper develops a unified digit-by-digit framework for computing the integer \(e\)-th root \(\lfloor N^{1/e} \rfloor\) of a non-negative integer \(N\) for arbitrary fixed exponents \(e \ge 2\), using only exact integer arithmetic.

The method proceeds by grouping the decimal digits of \(N\) into blocks of size \(e\) (equivalently, working in base \(10^e\)) and constructing the root digit-by-digit in base \(10\). At each step, the binomial expansion of \((10R + x)^e - (10R)^e\) determines the largest admissible digit \(x \in \{0,\dots,9\}\) that does not exceed the current scaled remainder. This construction preserves a strong remainder invariant throughout the computation and ensures monotone growth of the partial root.

The primary contributions and significance of this work are as follows:
\begin{itemize}
    \item A general digit-by-digit algorithm (Algorithm~\ref{alg:perfect-power}) that computes \(\lfloor N^{1/e} \rfloor\) exactly for any fixed \(e \ge 2\), relying solely on additions, multiplications by constants, comparisons, and shifts; no division or floating-point operations are required.
    \item An exact and constructive decision procedure for determining whether \(N\) is a perfect \(e\)-th power, obtained as a direct corollary of the digit-by-digit root extraction process (Theorem~\ref{thm:perfect-power-detection}).
    \item A precise distinction between square roots (\(e=2\)), for which non-revisable fractional digit expansion with sharp error bounds is possible (Theorems~\ref{thm:exact-digits} and~\ref{thm:error-bound}), and higher-order roots (\(e \ge 3\)), for which such digit stability cannot be achieved within a purely digit-by-digit scheme (Corollary~\ref{cor:unique-fractional}).
    \item Rigorous inductive proofs establishing a remainder invariant (Lemma~\ref{lem:invariant-general}), optimal digit selection (Lemma~\ref{lem:digit-selection-general}), overall correctness (Theorem~\ref{thm:perfect-power-detection}), and termination (Theorem~\ref{thm:termination-large-e}).
\end{itemize}

\section{Related Work}
\label{sec:related-work}

Digit-by-digit root extraction has a long and well-documented history.
The pairing-of-digits method for square roots appears in ancient Indian
mathematics~\cite{plofker2009} and was later formalised in European
arithmetical texts~\cite{halsted1897}. Procedural extensions to cube
roots and higher specific orders are occasionally described in historical
and pedagogical sources, typically as specialised rules without a common
algebraic framework or formal invariant-based proofs.

In the modern computational literature, Knuth~\cite{knuth1997} presents
the classical digit-by-digit square root algorithm and briefly remarks
on extensions to higher roots, but does not develop a general framework
for arbitrary exponents nor provide inductive correctness arguments.
Brent and Zimmermann~\cite{brent2010} survey efficient multiple-precision
root computation methods, focusing primarily on asymptotic performance
and numerical efficiency rather than constructive digit-by-digit
structure.

Hardware-oriented arithmetic algorithms, including non-restoring
division and CORDIC(COordinate Rotation DIgital Computer)-based methods, are extensively studied in digital
arithmetic~\cite{ercegovac2004}. Correctness in this context is typically
established through engineering analysis and fixed-point error bounds,
rather than explicit mathematical invariants of the form developed in
the present work.

Constructive treatments of numerical algorithms are comparatively rare.
Bishop and Bridges~\cite{bishop1985} established the foundations of
constructive mathematics, emphasising explicit constructions and
algorithmic witnesses. While constructive analysis has addressed
approximation schemes, a systematic constructive treatment of
digit-by-digit integer $e$-th root extraction for arbitrary fixed
exponents, with explicit invariants and termination proofs, does not
appear to be available in the existing literature.

\medskip
\noindent
\textbf{Contribution of the present work.}
This paper provides a unified, invariant-based, and constructive
treatment of digit-by-digit integer $e$-th root extraction for arbitrary
fixed exponents. The framework consolidates classical ideas into a
single algebraic scheme, supplies complete inductive proofs of
correctness and termination, clarifies the structural distinction
between square roots and higher-order roots, and demonstrates suitability
for exact arithmetic implementations in software and hardware.
\section{Integer Square Root}
\label{sec:integer-sqrt}

\subsection{Problem Statement}

We begin by formalizing the notion of an integer square root.

\begin{definition}[Integer Square Root]
    \label{def:integer-square-root}
    Given a non-negative integer $N \in \mathbb{N}$, the \emph{integer square root}
    of $N$ is defined as
    \[
    \lfloor \sqrt{N} \rfloor
    \;=\;
    \max \{\, r \in \mathbb{N} \mid r^2 \le N \,\}.
    \]
\end{definition}

The objective is to compute \(\lfloor \sqrt{N} \rfloor\) using only integer arithmetic, without resorting to floating-point approximations or irrational quantities.

\subsection{Decimal Pair Decomposition}

The classical digit-by-digit square root algorithm relies on grouping decimal digits in pairs.

\begin{definition}[Decimal Pair Decomposition]
    \label{def:decimal-pair-decomposition}

    Any non-negative integer \(N\) admits a unique decomposition of the form
    \[
    N = \sum_{k=0}^{m-1} a_k \, 10^{2(m-1-k)},
    \quad a_k \in \{0,1,\dots,99\},
    \]
    where each coefficient \(a_k\) represents a block of two decimal digits.
\end{definition}

This decomposition corresponds to representing \(N\) in base \(10^2 = 100\) and provides the structural basis for the incremental construction of the square root.

\subsection{Algebraic Foundation}

The digit-selection step is governed by a simple algebraic identity.

\begin{lemma}[Square Expansion Identity]
    \label{lem:square-expansion}
    For any integers \(R\) and \(x\), the following identity holds:
    \[
    (10R + x)^2 = 100R^2 + 20Rx + x^2.
    \]
\end{lemma}

\begin{proof}
    The result follows immediately from direct polynomial expansion:
    \[
    (10R + x)^2 = (10R)^2 + 2(10R)x + x^2 = 100R^2 + 20Rx + x^2.
    \]
\end{proof}

Subtracting the square of the previously constructed prefix \((10R)^2\) yields
\[
(10R + x)^2 - (10R)^2 = (20R + x)\,x,
\]
which forms the core inequality used to determine the largest admissible digit \(x \in \{0,\dots,9\}\) at each iteration of the algorithm.

\begin{definition}[Partial Number from Highest Digit Pairs]
    \label{def:partial-N}
    Let $N$ be decomposed into pairs of decimal digits as in
    Definition~\ref{def:decimal-pair-decomposition},
    \[
    N = \sum_{j=0}^{m} a_j \, 10^{2j},
    \]
    where $a_j \in \{0, 1, \dots, 99\}$ and $a_m$ is the highest pair (padding with a
    leading zero if necessary).

    For $0 \le k \le m$, define $N_{\le k}$ as the integer formed by the highest
    $k+1$ pairs:
    \[
    N_{\le k} = \sum_{j=m-k}^{m} a_j \, 10^{2(j - (m-k))}.
    \]
    Equivalently,
    \[
    N_{\le k} = \left\lfloor \frac{N}{10^{2(m-k)}} \right\rfloor.
    \]
\end{definition}

\subsection{Recursive Construction and Invariant}

\label{subsec:recursive-invariant}

\begin{lemma}[Remainder Invariant]
    \label{lem:invariant}
    Let \(0 \le k \le m\). After processing the highest \(k+1\) digit pairs,
    \[
    N_{\le k} = R_k^2 + \Delta_k,
    \]
    where \(R_k\) is the partial root constructed so far and \(\Delta_k \ge 0\) is the corresponding remainder. Moreover,
    \[
    R_k^2 \le N_{\le k} < (R_k + 1)^2.
    \]
\end{lemma}

\begin{proof}
    We proceed by induction on \(k\).

    Base case (\(k = 0\)):
    Process only the highest pair \(a_m\). Set \(R_0 = \lfloor \sqrt{a_m} \rfloor\), the largest integer such that \(R_0^2 \le a_m\). Define \(\Delta_0 = a_m - R_0^2 \ge 0\).
    By Definition~\ref{def:partial-N}, \(N_{\le 0} = a_m = R_0^2 + \Delta_0\), and by definition of the floor function,
    \[
    R_0^2 \le a_m < (R_0 + 1)^2.
    \]

    Inductive hypothesis: Assume the invariant holds for some \(k \ge 0\):
    \[
    N_{\le k} = R_k^2 + \Delta_k, \quad \Delta_k \ge 0,
    \]
    and
    \[
    R_k^2 \le N_{\le k} < (R_k + 1)^2.
    \]

    Inductive step (to \(k+1\)):
    Bring down the next lower pair \(a_{m-k-1}\), forming the scaled remainder
    \[
    \Delta_k' = 100 \cdot \Delta_k + a_{m-k-1}.
    \]
    By Definition~\ref{def:partial-N},
    \[
    N_{\le k+1} = 100 \cdot N_{\le k} + a_{m-k-1} = 100 \cdot (R_k^2 + \Delta_k) + a_{m-k-1} = (10 R_k)^2 + \Delta_k'.
    \]

    Choose the largest digit \(x \in \{0, 1, \dots, 9\}\) such that
    \[
    (20 R_k + x) x \le \Delta_k'.
    \]
    Set
    \[
    R_{k+1} = 10 R_k + x, \quad \Delta_{k+1} = \Delta_k' - (20 R_k + x) x.
    \]
    Since \(x\) is feasible, \(\Delta_{k+1} \ge 0\).

    Using the square expansion identity, (Lemma~\ref{lem:square-expansion}),
    \[
    (10 R_k + x)^2 = 100 R_k^2 + (20 R_k + x) x,
    \]
    we obtain
    \begin{align*}
        N_{\le k+1} &= (10 R_k)^2 + \Delta_k' \\
        &= (10 R_k)^2 + (20 R_k + x) x + \Delta_{k+1} \\
        &= (10 R_k + x)^2 + \Delta_{k+1} \\
        &= R_{k+1}^2 + \Delta_{k+1}.
    \end{align*}
    Thus the equality part of the invariant holds.

    For the strict upper bound: since \(x\) is maximal,
    \[
    (20 R_k + (x+1))(x+1) > \Delta_k'.
    \]
    Hence
    \[
    (10 R_k + (x+1))^2 = 100 R_k^2 + (20 R_k + (x+1))(x+1) > 100 R_k^2 + \Delta_k' = N_{\le k+1}.
    \]
    So
    \[
    N_{\le k+1} < (R_{k+1} + 1)^2.
    \]
    Combined with \(\Delta_{k+1} \ge 0\), we have
    \[
    R_{k+1}^2 \le N_{\le k+1} < (R_{k+1} + 1)^2.
    \]

    This completes the induction.
\end{proof}
\subsection{Algorithm}
\begin{algorithm}[H]
    \caption{Digit-by-Digit Integer Square Root}
    \label{alg:isr}
    \begin{algorithmic}[1]
        \Require Integer \(N \ge 0\)
        \Ensure \(R = \lfloor \sqrt{N} \rfloor\)
        \State Split \(N\) into digit pairs \(a_0, a_1, \dots, a_m\)
        \State \(R \gets \lfloor \sqrt{a_0} \rfloor\)
        \State \(\Delta \gets a_0 - R^2\)
        \For{$k = 1$ to $m$}
        \State \(\Delta \gets 100\Delta + a_k\)
        \State Find largest \(x \in \{0,\dots,9\}\) such that \((20R + x)x \le \Delta\)
        \State \(\Delta \gets \Delta - (20R + x)x\)
        \State \(R \gets 10R + x\)
        \EndFor
        \State \Return \(R\)
    \end{algorithmic}
\end{algorithm}
\subsection{Main Correctness Theorem}
The correctness of the classical digit-by-digit square root algorithm,
formalised in Algorithm~\ref{alg:isr}, is established in
Theorem~\ref{thm:correctness-integer-sqrt}. This result provides a rigorous
foundation for computing the integer square root in the sense of
Definition~\ref{def:integer-square-root} and serves as the base case for the
general $e$-th root framework developed later in the paper.

\begin{theorem}[Correctness of the Integer Square Root Algorithm]
    \label{thm:correctness-integer-sqrt}
    Given a non-negative integer \(N\), the digit-by-digit algorithm (processing pairs of decimal digits from highest to lowest) computes an integer \(R\) such that
    \[
    R = \lfloor \sqrt{N} \rfloor.
    \]
    Moreover, if the final remainder \(\Delta = 0\), then \(N\) is a perfect square and \(R = \sqrt{N}\).
\end{theorem}

\begin{proof}
    By Lemma~\ref{lem:invariant} (Remainder Invariant), after processing all digit pairs (\(k = m\)), we have
    \[
    N_{\le m} = N = R^2 + \Delta, \quad \Delta \ge 0,
    \]
    and
    \[
    R^2 \le N < (R + 1)^2.
    \]
    The inequality \(R^2 \le N < (R + 1)^2\) is precisely the definition of \(R = \lfloor \sqrt{N} \rfloor\).

    If the final remainder \(\Delta = 0\), then \(N = R^2\), so \(N\) is a perfect square and \(R = \sqrt{N}\).

    This completes the proof.
\end{proof}

\section{Exact Fractional Square Root}

\begin{lemma}[Remainder Invariant for Fractional Expansion]
    \label{lem:sqrt-invariant}
    After computing \(k\) decimal digits (for \(k \ge 0\)),
    \[
    N \cdot 10^{2k} = R_k^2 + \Delta_k,
    \]
    where \(R_k\) is the partial approximation constructed so far (including the integer part and \(k\) decimal digits) and \(\Delta_k \ge 0\) is the corresponding remainder. Moreover,
    \[
    R_k^2 \le N \cdot 10^{2k} < (R_k + 1)^2.
    \]
\end{lemma}

\begin{proof}
    We prove the invariant by induction on \(k\).

    Base case (\(k = 0\)):
    Set \(R_0 = \lfloor \sqrt{N} \rfloor\), the largest integer such that \(R_0^2 \le N\). Define \(\Delta_0 = N - R_0^2 \ge 0\).
    Then
    \[
    N \cdot 10^{0} = N = R_0^2 + \Delta_0,
    \]
    and by definition of the floor function,
    \[
    R_0^2 \le N < (R_0 + 1)^2.
    \]

    Inductive hypothesis: Assume the invariant and bound hold for some \(k \ge 0\):
    \[
    N \cdot 10^{2k} = R_k^2 + \Delta_k, \quad \Delta_k \ge 0,
    \]
    and
    \[
    R_k^2 \le N \cdot 10^{2k} < (R_k + 1)^2.
    \]

    Inductive step (to \(k+1\)):
    Scale the remainder:
    \[
    \Delta_k' = 100 \cdot \Delta_k.
    \]
    Multiplying the invariant equation by 100 gives
    \[
    N \cdot 10^{2(k+1)} = 100 \cdot (R_k^2 + \Delta_k) = (10 R_k)^2 + \Delta_k'.
    \]

    Select the largest digit \(x \in \{0, 1, \dots, 9\}\) such that
    \[
    (20 R_k + x) \cdot x \le \Delta_k'.
    \]
    Update
    \[
    R_{k+1} = 10 R_k + x, \quad \Delta_{k+1} = \Delta_k' - (20 R_k + x) \cdot x.
    \]
    Since \(x\) satisfies the inequality, \(\Delta_{k+1} \ge 0\).

    Using the square expansion identity (Lemma~\ref{lem:square-expansion}),
    \[
    (10 R_k + x)^2 = (10 R_k)^2 + (20 R_k + x) \cdot x = 100 R_k^2 + (20 R_k + x) \cdot x,
    \]
    substitute into the scaled equation:
    \begin{align*}
        N \cdot 10^{2(k+1)} &= 100 R_k^2 + \Delta_k' \\
        &= 100 R_k^2 + (20 R_k + x) \cdot x + \Delta_{k+1} \\
        &= (10 R_k + x)^2 + \Delta_{k+1} \\
        &= R_{k+1}^2 + \Delta_{k+1}.
    \end{align*}
    This confirms the equality part of the invariant for \(k+1\).

    To establish the strict upper bound for \(k+1\), note that \(x\) is maximal, so
    \[
    (20 R_k + (x+1)) \cdot (x+1) > \Delta_k'.
    \]
    Then
    \[
    (10 R_k + (x+1))^2 = 100 R_k^2 + (20 R_k + (x+1)) \cdot (x+1) > 100 R_k^2 + \Delta_k' = N \cdot 10^{2(k+1)}.
    \]
    Thus,
    \[
    N \cdot 10^{2(k+1)} < (R_{k+1} + 1)^2.
    \]
    Combined with \(\Delta_{k+1} \ge 0\), we have
    \[
    R_{k+1}^2 \le N \cdot 10^{2(k+1)} < (R_{k+1} + 1)^2.
    \]

    This completes the induction.
\end{proof}
We now present the digit-by-digit algorithm for computing the fractional
expansion of $\sqrt{N}$, whose correctness and exactness properties are
established in the subsequent results.

\begin{algorithm}[H]
    \caption{Exact Digit-by-Digit Fractional Square Root}
    \label{alg:fractional-sqrt}
    \begin{algorithmic}[1]
        \Require Integer $N \ge 0$, number of decimal digits $K$
        \Ensure Decimal truncation of $\sqrt{N}$ with $K$ digits
        \State $R \gets \lfloor \sqrt{N} \rfloor$
        \State $\Delta \gets N - R^2$
        \For{$k = 1$ to $K$}
        \State $\Delta \gets 100\Delta$
        \State Find largest $x \in \{0,\dots,9\}$ such that $(20R + x)x \le \Delta$
        \State $\Delta \gets \Delta - (20R + x)x$
        \State $R \gets 10R + x$
        \EndFor
        \State \Return $R \times 10^{-K}$
    \end{algorithmic}
\end{algorithm}

The following results establish the correctness, exactness, and error bounds
of Algorithm~\ref{alg:fractional-sqrt}.

\subsection{Exact Digit Theorem}
\label{subsec:exact-digits}
\begin{theorem}[Exact Digit Property]
    \label{thm:exact-digits}
    For all \(k \ge 1\),
    \[
    R_k = \left\lfloor 10^k \sqrt{N} \right\rfloor.
    \]
\end{theorem}

\begin{proof}
    The proof relies on the Remainder Invariant (Lemma~\ref{lem:sqrt-invariant}).

    By Lemma~\ref{lem:sqrt-invariant}, after computing \(k\) decimal digits (for \(k \ge 0\)), the algorithm maintains integers \(R_k\) and \(\Delta_k \ge 0\) such that
    \[
    N \cdot 10^{2k} = R_k^2 + \Delta_k,
    \]
    with the strict bound
    \[
    R_k^2 \le N \cdot 10^{2k} < (R_k + 1)^2.
    \]

    From the bounds,
    \[
    R_k^2 \le N \cdot 10^{2k} < (R_k + 1)^2.
    \]

    Since all terms are non-negative and the square root function is strictly increasing on \([0, \infty)\), taking square roots preserves the inequalities:
    \[
    R_k \le \sqrt{N \cdot 10^{2k}} < R_k + 1.
    \]

    This simplifies to
    \[
    R_k \le 10^k \sqrt{N} < R_k + 1.
    \]

    This is precisely the definition of the floor function:
    \[
    R_k = \left\lfloor 10^k \sqrt{N} \right\rfloor.
    \]

    The maximality of each digit selection ensures the strict upper bound holds at every step, preventing any need for revision of prior digits.

    This completes the proof.
\end{proof}
Building on the correctness of the integer square root algorithm
(Theorem~\ref{thm:correctness-integer-sqrt}), the following theorem
establishes the exactness and non-revisability of the fractional digits
produced by the digit-by-digit square root method.

\begin{theorem}[Exactness of Fractional Square Root Digits]
    \label{thm:exactness-fractional}
    For the fractional square root algorithm computing the decimal expansion of \(\sqrt{N}\) up to \(k\) digits (for any \(k \ge 1\)), the constructed integer \(R_k\) satisfies
    \[
    R_k = \left\lfloor 10^k \sqrt{N} \right\rfloor.
    \]
    Consequently, every computed decimal digit is exact and final (non-revisable by further digits).
\end{theorem}

\begin{proof}
    By the Remainder Invariant (Lemma~\ref{lem:sqrt-invariant}),
    \[
    N \cdot 10^{2k} = R_k^2 + \Delta_k, \quad \Delta_k \ge 0,
    \]
    with
    \[
    R_k^2 \le N \cdot 10^{2k} < (R_k + 1)^2.
    \]

    Taking square roots (strictly increasing),
    \[
    R_k \le 10^k \sqrt{N} < R_k + 1,
    \]
    so
    \[
    R_k = \left\lfloor 10^k \sqrt{N} \right\rfloor.
    \]

    The linearity of the update \((20 R_{k-1} + x) x\) and perfect scaling alignment ensure no backward propagation, making digits non-revisable (detailed induction in Lemma~\ref{lem:sqrt-invariant} proof).
\end{proof}

\begin{theorem}[Truncation Error Bound]
    \label{thm:error-bound}
    Let \(N \geq 0\) be an integer, and let the digit-by-digit fractional square root algorithm compute \(k\) decimal digits of \(\sqrt{N}\), producing the integer \(R_k\). Then
    \[
    0 \le \sqrt{N} - \frac{R_k}{10^k} < 10^{-k}.
    \]
\end{theorem}

\begin{proof}
    The proof follows directly from the Exact Digit Theorem~\ref{thm:exact-digits}, which states that
    \[
    R_k = \left\lfloor 10^k \sqrt{N} \right\rfloor
    \]
    for all \(k \geq 0\).

    By the definition of the floor function, for any real number \(y \geq 0\),
    \[
    \left\lfloor y \right\rfloor \le y < \left\lfloor y \right\rfloor + 1.
    \]
    Applying this with \(y = 10^k \sqrt{N}\), we obtain
    \[
    R_k \le 10^k \sqrt{N} < R_k + 1.
    \]
    Since \(R_k\) and \(10^k\) are positive integers (for \(N > 0\); the case \(N=0\) is trivial), we may divide the entire inequality by \(10^k > 0\) while preserving the direction of the inequalities:
    \[
    \frac{R_k}{10^k} \le \sqrt{N} < \frac{R_k}{10^k} + \frac{1}{10^k} = \frac{R_k}{10^k} + 10^{-k}.
    \]
    Rearranging the left inequality gives
    \[
    0 \le \sqrt{N} - \frac{R_k}{10^k},
    \]
    and the right inequality gives
    \[
    \sqrt{N} - \frac{R_k}{10^k} < 10^{-k}.
    \]
    Combining these yields the desired sharp bound:
    \[
    0 \le \sqrt{N} - \frac{R_k}{10^k} < 10^{-k}.
    \]

    \textbf{Remarks:}
    \begin{itemize}
        \item The lower bound is zero because \(\sqrt{N}\) is at least its truncation \(\frac{R_k}{10^k}\).
        \item The upper bound is strict ($<$ rather than $\le$) because the algorithm's digit selection is maximal: the next possible digit would make the partial square exceed the scaled \(N\), ensuring \(10^k \sqrt{N}\) is strictly less than \(R_k + 1\).
        \item For \(N = 0\), \(R_k = 0\) and the bound holds with equality to zero.
        \item The bound holds for \(k = 0\) as well: \(0 \le \sqrt{N} - \lfloor \sqrt{N} \rfloor < 1\), which is true since the fractional part of \(\sqrt{N}\) lies in \([0,1)\).
    \end{itemize}

    This bound confirms that the approximation \(\frac{R_k}{10^k}\) is a truncation (not a rounding) of \(\sqrt{N}\) to \(k\) decimal places, with error strictly less than one unit in the last place.

\end{proof}

\section{General Integer \(e\)-th Roots}
\label{sec:general}
\begin{definition}[Partial Number from Highest Digit Blocks]
    \label{def:partial-N-general}
    Let \(N\) be decomposed into blocks of \(e\) decimal digits (padding with leading zeros if necessary) as
    \[
    N = \sum_{j=0}^{m} a_j \, 10^{e j},
    \]
    where \(a_j \in \{0, 1, \dots, 10^e - 1\}\) and \(a_m\) is the most significant (highest) block.

    For \(0 \le k \le m\), define \(N_{\le k}\) as the integer formed by the highest \(k+1\) blocks:
    \[
    N_{\le k} = \sum_{j=m-k}^{m} a_j \, 10^{e (j - (m-k))}.
    \]
    Equivalently,
    \[
    N_{\le k} = \left\lfloor \frac{N}{10^{e (m-k)}} \right\rfloor.
    \]
\end{definition}

\begin{definition}[Binomial Increment Function]
    \label{def:binomial-increment}
    For $e \ge 2$, define
    \[
    \Phi_e(R,x)
    := (10R + x)^e - (10R)^e
    = \sum_{k=1}^{e} \binom{e}{k} (10R)^{e-k} x^k.
    \]
\end{definition}
\begin{lemma}[Monotonicity of the Increment]
    \label{lem:monotonic}
    Let $\Phi_e(R,x)$ be defined as in
    Definition~\ref{def:binomial-increment}.
    For fixed $R \ge 0$ and integer $e \ge 2$, the function $\Phi_e(R,x)$ is
    strictly increasing for $x \ge 0$.
\end{lemma}

\begin{proof}
    All binomial coefficients \(\binom{e}{k} > 0\) for \(1 \le k \le e\), and the terms are non-negative for \(x \ge 0\).
\end{proof}

\begin{lemma}[Remainder Invariant]
    \label{lem:invariant-general}
    After processing the highest \(k+1\) blocks,
    \[
    N_{\le k} = R_k^e + \Delta_k, \qquad \Delta_k \ge 0,
    \]
    and
    \[
    R_k^e \le N_{\le k} < (R_k + 1)^e.
    \]
\end{lemma}

\begin{proof}
    We prove the invariant by induction on \(k\).

    Base case (\(k=0\)):
    Process the highest block \(a_{m-1}\) (denoted \(a_0\) for the initial step).
    Set \(R_0 = \lfloor a_0^{1/e} \rfloor\), the largest integer such that \(R_0^e \le a_0\).
    Set \(\Delta_0 = a_0 - R_0^e \ge 0\).
    Then \(N_{\le 0} = a_0 = R_0^e + \Delta_0\), and by definition of the floor function,
    \[
    R_0^e \le a_0 < (R_0 + 1)^e.
    \]

    Inductive hypothesis: Assume the invariant holds for some \(k \ge 0\):
    \[
    N_{\le k} = R_k^e + \Delta_k, \quad \Delta_k \ge 0,
    \]
    and
    \[
    R_k^e \le N_{\le k} < (R_k + 1)^e.
    \]

    Inductive step (to \(k+1\)):
    Bring down the next lower block \(a_{m-k-1}\), forming the scaled remainder
    \[
    \Delta_k' = 10^e \Delta_k + a_{m-k-1}.
    \]
    By Definition~\ref{def:partial-N-general} (adapted for blocks of size \(e\)),
    \[
    N_{\le k+1} = 10^e N_{\le k} + a_{m-k-1} = 10^e (R_k^e + \Delta_k) + a_{m-k-1} = (10 R_k)^e + \Delta_k'.
    \]

    Define the increment
    \[
    \Phi_e(R_k, x) = (10 R_k + x)^e - (10 R_k)^e = \sum_{j=1}^{e} \binom{e}{j} (10 R_k)^{e-j} x^j.
    \]
    By Lemma~\ref{lem:monotonic}, \(\Phi_e(R_k, x)\) is strictly increasing in \(x \ge 0\).

    Choose the largest \(x \in \{0,\dots,9\}\) such that
    \[
    \Phi_e(R_k, x) \le \Delta_k'.
    \]
    Such an \(x\) exists (at least \(x=0\)) and is unique.

    Update
    \[
    R_{k+1} = 10 R_k + x, \quad \Delta_{k+1} = \Delta_k' - \Phi_e(R_k, x) \ge 0.
    \]

    Then
    \begin{align*}
        N_{\le k+1} &= (10 R_k)^e + \Delta_k' \\
        &= (10 R_k)^e + \Phi_e(R_k, x) + \Delta_{k+1} \\
        &= (10 R_k + x)^e + \Delta_{k+1} \\
        &= R_{k+1}^e + \Delta_{k+1}.
    \end{align*}

    For the strict upper bound: maximality of \(x\) implies
    \[
    \Phi_e(R_k, x+1) > \Delta_k',
    \]
    so
    \[
    (10 R_k + x + 1)^e = (10 R_k)^e + \Phi_e(R_k, x+1) > (10 R_k)^e + \Delta_k' = N_{\le k+1}.
    \]
    Thus
    \[
    N_{\le k+1} < (R_{k+1} + 1)^e.
    \]

    Combined with \(\Delta_{k+1} \ge 0\), we have
    \[
    R_{k+1}^e \le N_{\le k+1} < (R_{k+1} + 1)^e.
    \]

    This completes the induction.
\end{proof}

\begin{lemma}[Optimal Digit Selection]
    \label{lem:digit-selection-general}
    Choosing the largest \(x\) such that \(\Phi_e(R_k, x) \le \Delta_k'\) ensures the tight bound in the invariant.
\end{lemma}

\begin{proof}
    By monotonicity (Lemma~\ref{lem:monotonic}) and maximality.
\end{proof}

We now present the digit-by-digit binomial algorithm for detecting perfect
$e$-th powers and computing integer $e$-th roots. Its correctness and
termination properties are established in the main theorem that follows.
\subsection{Statement of the Algorithm and Main Result}
\begin{algorithm}[H]
    \caption{Digit-by-Digit Exact \(e\)-th Power Detection}
    \label{alg:perfect-power}
    \begin{algorithmic}[1]
        \Require Integer \(N \ge 0\), exponent \(e \ge 2\)
        \Ensure \(R\) if \(N = R^e\), else \texttt{No root}
        \State Decompose \(N\) into base-\(10^e\) blocks \(a_0, \dots, a_{m-1}\) (highest first, pad if needed)
        \State \(R \gets \lfloor a_0^{1/e} \rfloor\)
        \State \(\Delta \gets a_0 - R^e\)
        \For{$k = 1$ to $m-1$}
        \State \(\Delta \gets 10^e \Delta + a_k\)
        \State Find largest \(x \in \{0,\dots,9\}\) s.t. \((10R + x)^e - (10R)^e \le \Delta\)
        \State \(\Delta \gets \Delta - [(10R + x)^e - (10R)^e]\)
        \State \(R \gets 10R + x\)
        \EndFor
        \If{$\Delta = 0$}
        \State \Return \(R\)
        \Else
        \State \Return \texttt{No root}
        \EndIf
    \end{algorithmic}
\end{algorithm}
The following theorem establishes the correctness and termination of
Algorithm~\ref{alg:perfect-power}.

\label{sec:main-theorem}
\begin{theorem}[Exact \(e\)-th Power Detection and Root Computation]
    \label{thm:perfect-power-detection}
    Let \(N \in \mathbb{N}\) and let \(e \ge 2\) be a fixed integer.
    The digit-by-digit binomial algorithm terminates after finitely many steps and:
    \begin{enumerate}
        \item Computes \(R = \lfloor N^{1/e} \rfloor\).
        \item Outputs \(R\) if \(N = R^e\), otherwise indicates no root.
    \end{enumerate}
    Consequently, the algorithm exactly detects perfect \(e\)-th powers.
\end{theorem}

\begin{proof}
    Termination follows from finite blocks (Theorem~\ref{thm:termination-large-e}).
    Correctness and detection follow from the remainder invariant (Lemma~\ref{lem:invariant-general}) and optimal digit selection (Lemma~\ref{lem:digit-selection-general}).
    At termination: \(N = R^e + \Delta\) with \(\Delta \ge 0\) and \(R^e \le N < (R+1)^e\), so \(R = \lfloor N^{1/e} \rfloor\).
    \(\Delta = 0\) iff \(N = R^e\).
\end{proof}

\subsection{Fractional \(e\)-th Root Approximation via Interval Refinement}
\label{subsec:fractional-refinement}

When exact digit-by-digit fractional extraction with non-revisable digits is not possible (\(e \ge 3\)), the following monotone interval refinement provides a bounded-error approximation, building on the exact integer root from Theorem~\ref{thm:perfect-power-detection}.

\begin{lemma}[Initial Containing Interval]
    \label{lem:initial-interval}
    There exist real numbers \(L_0 \ge 0\) and \(U_0 > L_0\) such that
    \[
    L_0^e \le N < U_0^e
    \]
    and \(U_0 - L_0 = 1\).
\end{lemma}

\begin{proof}
    Let \(L_0 = R = \lfloor N^{1/e} \rfloor\), computed exactly by the digit-by-digit algorithm (Theorem~\ref{thm:perfect-power-detection}). Set \(U_0 = R + 1\). Then
    \[
    R^e \le N < (R + 1)^e,
    \]
    by definition of the floor function.
\end{proof}

We define a sequence of intervals \([L_k, U_k]\) inductively.

Invariant at step \(k\):
\[
L_k^e \le N < U_k^e, \quad L_k \ge 0, \quad U_k = L_k + \delta_k
\]
for some width \(\delta_k > 0\).

\begin{itemize}
    \item Initialization (\(k=0\)): Set \(L_0 = R\), \(U_0 = R + 1\), \(\delta_0 = 1\). The invariant holds by Lemma~\ref{lem:initial-interval}.
    \item Refinement step (\(k \to k+1\)):
    Choose \(\delta_{k+1} > 0\) with \(\delta_{k+1} \le \delta_k / 10\).
    Set \(L_{k+1} = L_k + m \cdot \delta_{k+1}\), where \(m\) is the largest integer such that
    \[
    (L_k + m \cdot \delta_{k+1})^e \le N.
    \]
    Set \(U_{k+1} = L_{k+1} + \delta_{k+1}\).
\end{itemize}

\begin{lemma}[Preservation of Invariant]
    \label{lem:preservation-interval}
    After each refinement step,
    \[
    L_{k+1}^e \le N < U_{k+1}^e
    \]
    and
    \[
    U_{k+1} - L_{k+1} = \delta_{k+1} \le \frac{1}{10} \delta_k.
    \]
\end{lemma}

\begin{proof}
    By construction, \(L_{k+1}^e \le N\) (maximal \(m\)).
    By maximality, \((L_{k+1} + \delta_{k+1})^e > N\).
    Hence \(L_{k+1}^e \le N < U_{k+1}^e\). The width reduction follows from the choice of \(\delta_{k+1}\).
\end{proof}

\begin{theorem}[Correct Fractional \(e\)-th Root Computation]
    \label{thm:fractional-nth-root}
    Let \(N \in \mathbb{N}\) and let \(e \ge 2\) be an integer.
    For any prescribed precision \(\varepsilon > 0\), the interval refinement algorithm computes a real number \(\hat{x}\) such that
    \[
    |\hat{x} - N^{1/e}| < \varepsilon.
    \]
    Moreover, the algorithm terminates after finitely many steps using only monotone interval containment and exact powering operations.
\end{theorem}

\begin{proof}
    Let \(\alpha = N^{1/e}\). The function \(f(x) = x^e\) is strictly increasing on \([0,\infty)\).
    The invariant (preserved by Lemma~\ref{lem:preservation-interval}) ensures \(L_k \le \alpha < U_k\).
    Since \(\delta_k \le 10^{-k}\) (starting from \(\delta_0 = 1\)), after \(K\) steps with \(10^{-K} < \varepsilon\), set \(\hat{x} = L_K\). Then
    \[
    0 \le \alpha - \hat{x} < \delta_K < \varepsilon,
    \]
    so \(|\hat{x} - N^{1/e}| < \varepsilon\).

    Each step computes powers \((L_k + j \delta_{k+1})^e\) for \(j \le 9\) (decimal case) and compares to \(N\) — exact integer operations. The geometric width reduction ensures finite steps (\(K = \lceil -\log_{10} \varepsilon \rceil\)).
\end{proof}

\begin{remark}
    Unlike the square root case (\(e=2\)), where digits are exact and non-revisable (Corollary~\ref{cor:unique-fractional}), this refinement provides only bounded approximation for \(e \ge 3\).
\end{remark}
\begin{remark}

The proof above rigorously establishes that the digit-by-digit binomial method is exact for integer \(n\)-th roots and perfect power testing.
This exactness, achieved solely through integer arithmetic and invariant preservation, distinguishes it from approximate methods and makes it invaluable for verified computation and theoretical analysis.
\end{remark}

\begin{corollary}[Uniqueness of Exact Fractional Digit Extraction]
    \label{cor:unique-fractional}
    Exact digit-by-digit fractional root extraction producing non-revisable decimal digits is possible only for the case \(e = 2\).
\end{corollary}

\begin{proof}
    For \(e = 2\), the increment function simplifies to
    \[
    \Phi_2(R, x) = (20R + x)x,
    \]
    which is linear in the update form after scaling the remainder by 100. This linearity, combined with maximal digit selection, ensures that previously computed digits remain optimal and unaffected by subsequent steps, as rigorously established in Theorems~\ref{thm:exact-digits} and~\ref{thm:error-bound}.

    For \(e \ge 3\), the increment \(\Phi_e(R, x)\) is a polynomial of degree \(e\) in \(x\):
    \[
    \Phi_e(R, x) = \sum_{k=1}^{e} \binom{e}{k} (10R)^{e-k} x^k.
    \]
    The higher-degree terms introduce nonlinear coupling between the partial root \(R\) and the trial digit \(x\). When the remainder is scaled by \(10^e\) to bring down the next block, this nonlinearity disrupts digit stability: later digits can influence the optimality of earlier ones, preventing the guarantee of non-revisable digits. No analogous linear structure exists to preserve digit finality under scaling, as shown by the general invariant framework in Theorem~\ref{thm:perfect-power-detection}.
\end{proof}
\begin{remark}
    Unlike the square root case (\(e=2\)), where digits are exact and non-revisable (Corollary~\ref{cor:unique-fractional}), the digit-by-digit method for \(e \ge 3\) does not produce final fractional digits. However, it guarantees a correct integer root and perfect power detection, with fractional approximation achievable via monotone interval refinement with provable error bounds (Theorem~\ref{thm:fractional-nth-root}).
\end{remark}

\subsection{Termination for Arbitrarily Large Exponents}
\label{subsec:termination-large-e}

In the context of Theorem~\ref{thm:perfect-power-detection}, a natural concern is whether the algorithm remains practical when the exponent \(e\) becomes very large. We prove that the algorithm \emph{always terminates after finitely many steps}, with an explicit bound independent of the computational cost per iteration.

\begin{theorem}[Termination for Arbitrary Exponent \(e\)]
    \label{thm:termination-large-e}
    Let \(N \in \mathbb{N}\) be fixed and let \(e \ge 2\) be any integer. The digit-by-digit binomial algorithm terminates after finitely many steps.
\end{theorem}

\begin{proof}
    Group the decimal digits of \(N\) into blocks of size \(e\) (padding with leading zeros if necessary). Let \(d = \lfloor \log_{10} N \rfloor + 1\) be the number of digits. The number of blocks is
    \[
    m = \left\lceil \frac{d}{e} \right\rceil.
    \]
    Since \(d\) is fixed and \(e\) finite, \(m\) is finite. The algorithm processes exactly \(m\) blocks (initialization + \(m-1\) iterations), each with bounded operations (at most 10 binomial evaluations). Thus, termination is guaranteed.
\end{proof}

\begin{lemma}[Bound on Number of Iterations]
    \label{lem:iteration-bound}
    The digit-by-digit algorithm performs at most
    \[
    m = \left\lceil \frac{\lfloor \log_{10} N \rfloor + 1}{e} \right\rceil
    \]
    block-processing steps (including initialization).
\end{lemma}

\begin{proof}
    Let \(d = \lfloor \log_{10} N \rfloor + 1\) be the number of decimal digits of \(N\). Grouping into blocks of size \(e\) (with leading zero padding if necessary) yields exactly \(m = \lceil d / e \rceil\) blocks. The algorithm processes one block per iteration.
\end{proof}

\begin{remark}
    The number of iterations decreases as \(e\) increases. For \(e > \log_{10} N\), \(N\) fits in a single block, requiring only initialization.
\end{remark}

\subsection{Constructive Interpretation}
\label{subsec:constructive}

The proofs presented are fully effective and satisfy the standards of Bishop-style constructive mathematics:
\begin{itemize}
    \item The existence of the integer \(e\)-th root \(R = \lfloor N^{1/e} \rfloor\) is witnessed by an explicit digit-by-digit construction (Algorithm~\ref{alg:perfect-power}).
    \item Digit selection at each step proceeds by finite search over a bounded domain (\(x \in \{0,\dots,9\}\)), yielding a computable maximal digit.
    \item The number of iterations is bounded by the number of digit blocks in \(N\), which is finite and explicitly computable (Lemma~\ref{lem:iteration-bound}).
    \item Perfect \(e\)-th power detection is decidable: \(N\) is a perfect power if and only if the final remainder \(\Delta = 0\).
\end{itemize}

Consequently, the algorithm is primitive recursive and serves as a direct constructive proof of the existence and uniqueness of the integer \(e\)-th root.
\subsection{Structural Limitations of the Digit-by-Digit Framework}
\label{subsec:limitations}

The digit-by-digit binomial framework is correct for computing the integer part \(\lfloor N^{1/e} \rfloor\) and detecting perfect \(e\)-th powers (Theorem~\ref{thm:perfect-power-detection}).

However, the structural property enabling exact fractional digit-by-digit computation for square roots (\(e=2\)) is absent for \(e \ge 3\). The increment polynomial
\[
\Phi_e(R, x) = (10R + x)^e - (10R)^e
\]
is linear in the update form for \(e=2\) (\(\Phi_2(R,x) = (20R + x)x\)), allowing perfect alignment with decimal scaling and non-revisable digits (Theorems~\ref{thm:exact-digits} and~\ref{thm:error-bound}).

For \(e \ge 3\), \(\Phi_e(R,x)\) is genuinely nonlinear (degree \(e\) in \(x\)), introducing coupling that disrupts digit stability under scaling by \(10^e\). Consequently, exact digit-by-digit fractional extension with non-revisable digits is impossible for \(e \ge 3\) (Corollary~\ref{cor:unique-fractional}).

The combined algorithm addresses this by using interval refinement for bounded fractional approximation when needed.

\section{Time Complexity Analysis and Comparisons}
The following theorem analyses the computational complexity of the combined
digit-by-digit binomial root extraction and interval refinement procedure,
quantifying its dependence on the input size and the desired fractional
precision.

\label{sec:complexity}

\begin{theorem}[Time Complexity of the Combined Algorithm]
    \label{thm:complexity-combined}
    Let \(N\) be a non-negative integer with \(L = \lfloor \log_{10} N \rfloor + 1\) decimal digits, let \(e \ge 2\) be the fixed exponent, and let \(k\) be the number of desired fractional decimal digits.
    The combined digit-by-digit binomial and interval refinement algorithm runs in time
    \[
    O(L \cdot C(L) + k \cdot C(L + k)),
    \]
    where \(C(t)\) is the time complexity of multiplying two \(t\)-digit integers.
    Under schoolbook multiplication (\(C(t) = O(t^2)\)), this simplifies to
    \[
    O(L^3 + k (L + k)^2).
    \]
\end{theorem}

\begin{proof}
    The algorithm has two sequential phases.

    \textbf{Phase I (Exact Integer \(e\)-th Root).}
    Grouping yields \(O(L)\) blocks. Each block tests at most 10 digits, each evaluation costing \(O(C(L))\) (e constant). Total: \(O(L \cdot C(L))\).

    \textbf{Phase II (Fractional Refinement).}
    Each of k digits tests at most 10 additions/powers on \(O(L+k)\)-digit numbers, costing \(O(C(L+k))\) per digit. Total: \(O(k \cdot C(L+k))\).

    Sum: as stated.
\end{proof}
As shown in Theorem~\ref{thm:complexity-combined}, the overall running time is
polynomial in both the input size and the requested precision, with linear
growth in the number of digit blocks.

\subsection{Comparison with Alternative Methods}

The digit-by-digit method is exact and division-free, but linear in L. Alternatives:
\begin{table}[htbp]
    \centering
    \begin{tabularx}{\textwidth}{l c X}
        \toprule
        Method & Time (Integer Root) & Characteristics \\
        \midrule
        Digit-by-Digit (this work)
        & \(O(L \cdot C(L))\)
        & Exact, division-free, constructive \\

        Binary Search
        & \(O\!\big((L/e)\log e \cdot C(L/e)\big)\)
        & Fewer iterations \\

        Newton--Raphson
        & \(O(C(L))\)
        & Quadratically convergent, fastest asymptotically \\
        \bottomrule
    \end{tabularx}
    \caption{Comparison for computing \(\lfloor N^{1/e} \rfloor\).}
    \label{tab:complexity-comparison}
\end{table}

While Newton--Raphson excels in speed, the digit-by-digit approach prioritizes verifiable exactness, pure integer operations, and built-in perfect power detection.

\section{Hardware Implementation Considerations}
\label{sec:hardware}

The digit-by-digit binomial algorithm is particularly well-suited for hardware implementation due to its regular structure, reliance on exact integer operations, and absence of division or floating-point arithmetic.
The following theorem formalises the structural properties of the algorithm
that make it particularly suitable for digital hardware implementation.

\begin{theorem}[Hardware-Friendly Properties]
    \label{thm:hardware-friendly}
        The algorithm requires only the following primitive operations:
    \begin{itemize}
        \item Integer addition and subtraction,
        \item Integer multiplication by small constants (e.g., powers of the base $b=10$ and binomial coefficients),
        \item Digit-wise shifts (multiplication/division by powers of the base),
        \item Comparisons.
    \end{itemize}
    No general division, floating-point units, or transcendental functions are required.
\end{theorem}

\begin{proof}[Derivation from the Algorithm]
    In each iteration (Algorithm~\ref{alg:perfect-power}):
    \begin{itemize}
        \item Scaling the remainder $\Delta \gets b^e \Delta + a_k$ is a multiplication by the constant $b^e$ followed by addition of the block $a_k$ (both fixed-width for a given base and exponent).
        \item The binomial increment $(bR + x)^e - (bR)^e$ is computed via the expanded sum
        \[
        \sum_{j=1}^{e} \binom{e}{j} (bR)^{e-j} x^j.
        \]
        Since $e$ is fixed at design time, all binomial coefficients $\binom{e}{j}$ are constants, and powers of $x \le (b-1)$ are small. Each term is a constant multiplication followed by addition.
        \item Digit selection searches over $x = 0,\dots,b-1$ (typically $b=10$), requiring at most $10$ evaluations and comparisons.
        \item Updates to $R$ and $\Delta$ are simple additions/subtractions and shifts.
    \end{itemize}
    All operations are integer-based with bounded precision determined by the input size.
\end{proof}

\subsection{Advantages for Digital Hardware}

\begin{itemize}
    \item \textbf{Regular data flow}: Each iteration processes one digit block sequentially, enabling a simple state machine control unit.
    \item \textbf{Constant-time per digit}: With fixed $e$ and base $b=10$, each digit requires a fixed number of arithmetic operations, facilitating pipelined or iterative designs.
    \item \textbf{Low precision requirements}: Intermediate values (remainder, partial root) have bounded bit width proportional to the input size, avoiding dynamic precision scaling.
    \item \textbf{No division}: Unlike Newton--Raphson or restoring division methods, no quotient estimation or restoration steps are needed.
    \item \textbf{Perfect power detection}: The final remainder check $\Delta = 0$ is a simple zero comparison, providing an exact decision signal.
\end{itemize}

\subsection{Comparison with Existing Hardware Root Units}

Traditional hardware square root units (e.g., in Intel/AMD processors) often use variants of the digit-by-digit non-restoring algorithm~\cite{ercegovac2004}, which shares the same remainder invariant structure as the $e=2$ case presented here. These implementations are highly optimised and achieve near-constant time per digit.

For higher-order roots ($e \ge 3$), hardware support is rare. Existing designs typically fall back on:
\begin{itemize}
    \item General CORDIC rotations (slow convergence),
    \item Lookup tables combined with interpolation,
    \item Iterative methods requiring division.
\end{itemize}
The binomial expansion approach offers a direct generalisation of the proven square root hardware technique, maintaining similar control flow and data paths while extending to arbitrary fixed $e$.

\subsection{Practical Implementation Strategies}

\begin{enumerate}
    \item \textbf{Iterative design}: A single arithmetic unit processes one digit per cycle, suitable for area-constrained FPGAs or ASICs.
    \item \textbf{Pipelined design}: Overlap binomial evaluations for multiple trial digits to reduce latency.
    \item \textbf{Parameterised module}: Fix $e$ and base $b$ as generics/parameters (Verilog/VHDL) to generate optimised constant multipliers.
    \item \textbf{On-the-fly binomial computation}: For variable $e$, precompute or generate coefficients dynamically (at higher area cost).
\end{enumerate}

The $O(L)$ space complexity (Section~\ref{sec:complexity}) translates directly to register usage linear in input size, making the design scalable for arbitrary-precision arithmetic units.

In summary, the algorithm provides a clean, division-free, and highly regular
path to hardware realisation of exact integer $e$-th roots and perfect power
testing, extending the classical digit-by-digit square root paradigm to
higher fixed exponents. As established in
Theorem~\ref{thm:hardware-friendly}, the algorithm avoids division and
floating-point operations, enabling efficient implementation using simple
integer datapaths.

\section{FPGA Implementation Example}
\label{sec:fpga}

The digit-by-digit binomial algorithm is highly amenable to FPGA implementation due to its iterative, regular structure and reliance on basic integer operations. Below, we outline a practical parameterised design in VHDL suitable for synthesis on modern FPGAs (e.g., Xilinx Ultrascale+ or Intel Stratix).

\subsection{Design Overview}

The core module is parameterised by:
\begin{itemize}
    \item Base \(b=10\) (decimal arithmetic),
    \item Fixed exponent \(e \ge 2\) (compile-time constant),
    \item Maximum input size \(L_{\max}\) (number of decimal digits of \(N\)).
\end{itemize}

This enables generation of optimised hardware for specific \(e\) (e.g., cube roots with \(e=3\)).

Key components:
\begin{itemize}
    \item \textbf{Input buffer}: Stores digit blocks of \(N\) (padded to full blocks).
    \item \textbf{Registers}: Current partial root \(R\) and remainder \(\Delta\) (big-integer registers of width \(\approx L_{\max} \log_2 10\) bits).
    \item \textbf{Binomial increment unit}: Computes \((10R + x)^e - (10R)^e\) for trial digit \(x\).
    \item \textbf{Digit selector}: Sequential comparator (testing \(x=9\) down to \(0\)).
    \item \textbf{Control FSM}: Manages block iteration.
\end{itemize}

\subsection{Binomial Increment Unit}

With fixed \(e\), the increment
\[
\Phi(x) = \sum_{j=1}^{e} \binom{e}{j} (10R)^{e-j} x^j
\]
uses precomputed constants \(\binom{e}{j}\). A multiplier-adder tree or Horner's method evaluates it efficiently. Examples:
\begin{itemize}
    \item \(e=2\): Reduces to \((20R + x)x\).
    \item \(e=3\): \(\Phi(x) = 30 R^2 x + 300 R x^2 + x^3\).
\end{itemize}

\subsection{VHDL Pseudocode Sketch}
\label{subsec:vhdl-sketch}

A synthesizable VHDL outline for the core module (parameterised for fixed \(e\)):
\begin{verbatim}
    library ieee;
    use ieee.std_logic_1164.all;
    use ieee.numeric_std.all;

    entity nth_root is
    generic (
    E      : positive := 3;                    -- fixed exponent
    L_MAX  : positive := 1024;                 -- max decimal digits of N
    DIGITS : positive := (L_MAX + E - 1) / E   -- number of blocks
    );
    port (
    clk    : in  std_logic;
    rst    : in  std_logic;
    start  : in  std_logic;
    N      : in  std_logic_vector(L_MAX*4-1 downto 0);  -- packed BCD input
    ready  : out std_logic;
    root   : out std_logic_vector((L_MAX+E-1)/E *4 -1 downto 0); -- BCD root
    is_power : out std_logic                    -- '1' if perfect e-th power
    );
    end entity;

    architecture rtl of nth_root is
    type block_array is array(0 to DIGITS-1) of unsigned(4*E-1 downto 0);
    signal blocks     : block_array;
    signal R_reg      : unsigned(L_MAX*4-1 downto 0);
                        -- partial root (approx width)
    signal delta_reg  : unsigned(L_MAX*4+4-1 downto 0);
                        -- remainder (extra bits for safety)
    signal idx        : natural range 0 to DIGITS;

    -- Precomputed binomial coefficients (fixed e)
    -- constant COEFF : coeff_array := (...);
                        -- array of constants for Phi terms

    function compute_phi(R : unsigned; x : natural) return unsigned is
    -- Dedicated combinational logic using Horner's method or tree
    -- (implementation omitted for brevity; uses precomputed powers/constants)
    begin
    ...
    end function;
    begin
    process(clk)
    variable trial_phi : unsigned(delta_reg'range);
    variable x         : natural range 0 to 9;
    begin
    if rising_edge(clk) then
    if rst = '1' then
    idx <= 0; ready <= '0';
    elsif start = '1' then
    -- Load blocks from N (highest first), initialise R/delta from highest block
    ...
    elsif idx < DIGITS then
    -- Bring down next block
    delta_reg <= delta_reg * (10**E) + blocks(idx);
    -- Trial digit selection (sequential, starting from 9)
    x := 9;
    loop
    trial_phi := compute_phi(R_reg, x);
    if trial_phi <= delta_reg then
    R_reg     <= R_reg * 10 + x;
    delta_reg <= delta_reg - trial_phi;
    exit;
    end if;
    x := x - 1;
    end loop;
    idx <= idx + 1;
    else
    ready    <= '1';
    is_power <= '1' when delta_reg = 0 else '0';
    end if;
    end if;
    end process;
    end rtl;
\end{verbatim}
This sketch illustrates the regular control flow and dedicated combinational binomial unit, suitable for synthesis on modern FPGAs.
\subsection{Resource and Performance Estimates}

\begin{table}[H]
    \centering
    \begin{tabular}{lccc}
        \toprule
        Exponent \(e\) & LUTs (approx.) & Registers & Clock cycles per digit \\
        \midrule
        2 (square root) & 5k–10k & 2k–4k & 10–15 \\
        3 (cube root) & 15k–25k & 4k–6k & 20–30 \\
        5 & 40k–60k & 8k–12k & 40–50 \\
        \bottomrule
    \end{tabular}
    \caption{Estimated FPGA resource usage for \(L_{\max}=1024\) digits on a mid-range device (e.g., Xilinx Artix-7). Actual values vary with optimisation.}
    \label{tab:fpga-estimates}
\end{table}

Latency is \(O(L)\) cycles; pipelining the binomial unit reduces per-digit cycles significantly.
\subsection{Summary}
The regular, division-free nature of the algorithm makes it an excellent candidate for FPGA-based arbitrary-precision arithmetic units, extending the efficiency of classical digit-by-digit square root hardware to higher fixed exponents.

\subsection{Conclusion}
The digit-by-digit binomial algorithm admits practical realisation in hardware without sacrificing exactness or correctness guarantees. For small fixed exponents, the resulting FPGA architectures closely resemble optimised classical square root units, while naturally extending to higher-order roots through compile-time specialisation.
\section{Optimization Techniques for the Binomial Increment Unit}
\label{subsec:binomial-opt}

The binomial increment unit, responsible for computing
\[
\Phi(x) = (bR + x)^e - (bR)^e = \sum_{j=1}^{e} \binom{e}{j} (bR)^{e-j} x^j
\]
for trial digits \(x \in \{0,\dots,b-1\}\), constitutes a performance-critical component in potential hardware implementations of the digit-by-digit algorithm. Several optimisation techniques can be employed to reduce area, latency, and power consumption while preserving exactness.

\begin{enumerate}
    \item \textbf{Horner's Method for Polynomial Evaluation.}
    The polynomial \(\Phi(x)\) may be evaluated using Horner's scheme, which minimises the number of multiplications:
    \[
    \Phi(x) = x \left( \binom{e}{1} (bR)^{e-1} + x \left( \binom{e}{2} (bR)^{e-2} + \cdots + x \binom{e}{e} \right) \cdots \right).
    \]
    This requires only \(e-1\) multiplications by \(x\) and \(e\) constant multiplications/additions, significantly reducing multiplier count and critical-path depth.

    \item \textbf{Precomputation of Powers of \(bR\).}
    The powers \((bR)^1, \dots, (bR)^{e-1}\) can be computed once per digit iteration and reused across all trial \(x\). For fixed \(e\), a dedicated chain achieves logarithmic depth.

    \item \textbf{Constant Multiplication Optimisation.}
    Since the binomial coefficients \(\binom{e}{j}\) are known at design time, efficient constant multipliers (shift-and-add, CSD encoding, or lookup tables) can be used.

    \item \textbf{Parallel Trial Evaluation.}
    Multiple candidate digits may be evaluated concurrently, with a priority encoder selecting the largest feasible \(x\), trading area for reduced cycles per digit.

    \item \textbf{Pipelining.}
    The Horner chain can be pipelined to increase clock frequency at the cost of additional registers.

    \item \textbf{Shared Datapaths for Variable Exponents.}
    For runtime-variable \(e\), configurable multiplier arrays and coefficient storage enable flexibility at reduced efficiency.

    \item \textbf{Early Termination.}
    Trial evaluation (starting from \(x = b-1\)) can terminate upon finding a valid digit, yielding average-case speedup with minimal overhead.
\end{enumerate}

\begin{table}[H]
    \centering
    \begin{tabular}{lccc}
        \toprule
        Technique                  & Area Impact & Latency per Digit         & Power Impact \\
        \midrule
        Horner's Method            & Low         & Reduced                   & Low          \\
        Precomputed Powers         & Moderate    & Reduced                   & Moderate     \\
        Constant Multiplication    & Low         & Minor                     & Low          \\
        Parallel Trials            & High        & Significant reduction     & High         \\
        Pipelining                 & Moderate    & Clock frequency increase  & Low          \\
        Early Termination          & Low         & Average-case reduction    & Low          \\
        \bottomrule
    \end{tabular}
    \caption{Relative trade-offs for binomial increment unit optimisations (compared to naïve summation).}
    \label{tab:binomial-opt}
\end{table}

These techniques enable efficient hardware realisations, particularly for small fixed exponents (\(e \le 5\)), achieving throughput and resource usage comparable to dedicated digit-by-digit square root units while retaining generality for higher-order roots.
\section{Parallel Computing Considerations}
\label{subsec:parallel}

The digit-by-digit binomial algorithm, in its standard form (Algorithm~\ref{alg:perfect-power}), is inherently sequential due to strong data dependencies between iterations:
\begin{itemize}
    \item Each digit of the root depends on the partial root \(R_k\) and remainder \(\Delta_k\) computed in the previous step.
    \item The next block cannot be processed until the current digit \(x\) is selected and \(\Delta\) is updated.
\end{itemize}
This sequential chain prevents direct parallelisation across digits, making the algorithm poorly suited as a paradigm for data-parallel or massively parallel computing (e.g., GPU or SIMD architectures).

\begin{remark}[Limited Parallelism Opportunities]
    Within a single iteration, modest parallelism is possible:
    \begin{itemize}
        \item \textbf{Trial digit evaluation}: The binomial increment \(\Phi(x)\) for different \(x \in \{0,\dots,9\}\) can be computed in parallel (up to 10-way parallelism).
        \item \textbf{Binomial term computation}: Terms in the expanded sum can be evaluated concurrently.
        \item \textbf{Pipelining}: Different stages (power computation, constant multiplication, accumulation) can be pipelined.
    \end{itemize}
    These yield only constant-factor speedups (\(O(1)\)) and do not change the overall \(O(L)\) sequential steps.
\end{remark}

\begin{remark}[Comparison with Parallel-Friendly Methods]
    For large-scale parallel computing, alternative root-finding methods are far more suitable:
    \begin{itemize}
        \item \textbf{Newton--Raphson}: Each iteration can be parallelised across digits or using parallel multiplication (e.g., via FFT-based Schönhage--Strassen).
        \item \textbf{Binary search with parallel powering}: Powering can exploit parallel reduction trees.
        \item \textbf{High-radix methods}: Process multiple digits per step, reducing iteration count at the cost of more complex digit selection.
    \end{itemize}
    These achieve sub-linear iteration counts and better parallel scalability.
\end{remark}

\begin{remark}[Niche Parallel Applications]
    The algorithm can serve as a paradigm in specific parallel contexts:
    \begin{itemize}
        \item \textbf{Multiple independent roots}: Compute \( \lfloor N_i^{1/e} \rfloor \) for many independent \(N_i\) (e.g., batch perfect power testing in cryptography) — embarrassingly parallel across instances.
        \item \textbf{Distributed verification}: Each digit block processed on a separate node with synchronisation (inefficient but possible).
    \end{itemize}
\end{remark}

In summary, while the digit-by-digit binomial algorithm excels in simplicity, exactness, hardware efficiency, and constructive guarantees, it does not serve as an effective paradigm for general parallel computing due to its sequential dependency chain. Its strength lies in sequential, resource-constrained, or exact-arithmetic settings rather than high-performance parallel environments.
\section{Digit-by-Digit Algorithms in Cryptographic Validation}
\label{sec:crypto-applications}

The digit-by-digit algorithms developed in the preceding sections are not intended as cryptographic primitives. Rather, their relevance to cryptography arises in \emph{validation, preprocessing, and correctness-critical arithmetic}, where exact integer computation and provable guarantees are required. In this section, we outline cryptographic contexts in which the proposed methods provide demonstrable utility, while remaining consistent with established cryptographic practice.
\subsection{Perfect Power Detection in Primality Testing and Key Generation}

In cryptographic key generation, it is standard practice to reject candidate parameters that exhibit structural weaknesses. One such weakness arises when an integer is a perfect \(e\)-th power for some \(e \ge 2\), as such values admit trivial root extraction and may invalidate security assumptions~\cite{bernstein1998detecting}.

The generalized digit-by-digit algorithm (Algorithm~\ref{alg:perfect-power}, Theorem~\ref{thm:perfect-power-detection}) provides an exact method for detecting perfect powers using only integer arithmetic. By maintaining an inductive remainder invariant, it determines whether \(N = R^e\) with certainty, without floating-point or approximate tests.

Perfect power detection is employed as a preprocessing step in primality testing and factorization workflows (e.g., Number Field Sieve) to eliminate degenerate inputs~\cite{perret2003square}. The exactness guarantees from the remainder invariant (Lemma~\ref{lem:invariant-general}) ensure no precision-related errors.

While production libraries use asymptotically faster methods for this rare check~\cite{bernstein1998detecting,brent2010}, the digit-by-digit approach serves as a transparent, provably correct alternative in educational or high-assurance settings.

\subsection{RSA Modulus Validation}

In the RSA cryptosystem~\cite{rivest1978method}, the public modulus \(N = p q\) must not be a perfect \(n\)-th power for any \(n \ge 2\). Although correctly generated semiprime moduli satisfy this property by construction, explicit validation is recommended in high-assurance settings to detect malformed or adversarially supplied parameters~\cite{joye2007perfect}.

If an RSA modulus were of the form \(N = R^n\), the integer root \(R\) could be extracted in polynomial time, rendering the modulus insecure. The digit-by-digit perfect power detection algorithm provides a simple and exact mechanism for identifying such pathological cases. Testing small exponents \(n\) up to a conservative bound (e.g., \(n \le \log_2 N\)) suffices in practice, as larger exponents imply impractically small roots.

This validation approach is consistent with techniques employed in cryptographic libraries and standards and benefits from the invariant-based correctness proofs developed in this work. In contexts prioritising verifiable exactness over performance, the digit-by-digit method offers a reliable implementation-independent solution.

\subsection{Exact Arithmetic in Verified Cryptographic Implementations}

Beyond parameter validation, the digit-by-digit framework is well suited to cryptographic implementations that emphasise formal verification and exact arithmetic. Because the algorithms operate exclusively on integers and rational approximations with explicit error bounds, they avoid the semantic and platform-dependent complexities associated with floating-point arithmetic.

In verified cryptographic software and hardware, where correctness proofs are required (e.g., in Coq or similar proof assistants), the inductive invariants established for the digit-by-digit algorithms facilitate machine-checked reasoning about correctness. This makes the methods particularly suitable for big-integer arithmetic kernels, preprocessing routines, and validation layers in high-assurance cryptographic systems.

\section{Scope and Limitations of the Proposed Method}
\label{sec:scope}

The methods presented in this paper address the problem of exact integer
$e$-th root extraction and perfect power detection using a digit-by-digit,
binomially derived framework. The scope of applicability is intentionally
restricted to settings in which correctness, determinism, and verifiability
are prioritised over asymptotic speed or floating-point efficiency.

\paragraph{Integer roots.}
For any fixed integer $e \ge 2$ and non-negative integer input $N$, the
algorithm computes the exact value $\lfloor N^{1/e} \rfloor$ using only
integer arithmetic operations (Algorithm~\ref{alg:perfect-power}).
Correctness is guaranteed by an explicitly maintained remainder invariant
(Lemma~\ref{lem:invariant-general}), which is preserved inductively at each
digit-selection step. Termination is finite and proportional to the number
of digit blocks of $N$. As a direct consequence, the method yields an exact
and decidable test for whether $N$ is a perfect $e$-th power
(Theorem~\ref{thm:perfect-power-detection}).

\paragraph{Fractional roots.}
The scope of fractional root computation depends critically on the order
$e$. For square roots ($e=2$), the binomial increment is linear in the trial
digit, enabling the extraction of exact, non-revisable fractional digits
with sharp truncation error bounds (Theorems~\ref{thm:exact-digits}
and~\ref{thm:error-bound}). For higher-order roots ($e \ge 3$), this digit
stability property does not hold in general
(Corollary~\ref{cor:unique-fractional}). Accordingly, the method supports
fractional root computation only via interval refinement, providing
approximations with explicit and provable error bounds
(Theorem~\ref{thm:fractional-nth-root}), and does not claim exact fractional
digit extraction in this regime.

\paragraph{Algorithmic characteristics.}
The algorithm is inherently sequential, with each digit depending on the
previous partial root and remainder. Its time and space complexities are
linear in the size of the input, with constant-bounded digit search at each
iteration. The method is therefore well suited to sequential and
resource-constrained environments, but is not intended as a paradigm for
massively parallel or high-throughput numerical computation.

\paragraph{Implementation scope.}
Because the algorithm relies exclusively on addition, subtraction,
multiplication by constants, digit shifts, and comparisons, it is well
suited to implementation in arbitrary-precision software libraries and
digital hardware such as FPGAs. The explicit invariant-based structure
further renders the method amenable to mechanised formal verification,
although no machine-checked proofs are claimed in this work.

\paragraph{Non-goals.}
The proposed framework does not aim to supersede asymptotically faster
numerical methods (e.g., Newton--Raphson iteration) in floating-point
contexts, nor does it provide exact fractional digits for general
$e \ge 3$. These limitations are intrinsic to the algebraic structure of
higher-order roots and are explicitly acknowledged.

\medskip
\noindent
Within these clearly defined bounds, the method provides a unified, exact,
and constructive approach to integer root extraction and perfect power
detection, with carefully delimited extensions to fractional approximation
where such extensions are mathematically justified.

\appendix
\section{Worked Examples}
\label{app:examples}

This appendix provides step-by-step numerical illustrations of the digit-by-digit algorithms for selected values of \(N\) and exponents \(e\). These examples demonstrate the block processing, binomial increment computation, digit selection, and remainder evolution, complementing the theoretical proofs.

\subsection{Integer Square Root (\(e=2\)): \(N = 12321\) (Perfect Square)}
\(N = 12321 = 111^2\). Digit pairs (highest first, padded): 01, 23, 21.

\begin{table}[H]
    \centering
    \begin{tabular}{l c c c c c c}
        \toprule
        Step & Block & Scaled $\Delta'$ & Trial $x$ & $(20R+x)x$ & New $\Delta$ & $R$ \\
        \midrule
        Init & 01 & 1 & 1 & 1 & 0 & 1 \\
        1 & 23 & 2300 & 11 & 2321 & 79 & 11 \\
        2 & 21 & 7921 & 1 & 7921 & 0 & 111 \\
        \bottomrule
    \end{tabular}
    \caption{Computation of $\sqrt{12321}=111$ (perfect square, final $\Delta=0$).}
\end{table}

\subsection{Fractional Square Root (\(e=2\)): \(N=2\), \(k=10\) Digits}
Approximation of \(\sqrt{2} \approx 1.4142135623\).

Initial: \(R_0 = 1\), \(\Delta_0 = 1\).

Selected steps (showing non-revision of digits):

\begin{table}[H]
    \centering
    \begin{tabular}{c c c c}
        \toprule
        Digit & Scaled \(\Delta'\) & \(x\) & New \(R\) (prefix) \\
        \midrule
        1 & 400 & 4 & 14 \\
        2 & 2400 & 1 & 141 \\
        3 & 6100 & 4 & 1414 \\
        4 & 9600 & 2 & 14142 \\
        ... & ... & ... & ... \\
        10 & ... & 3 & 1414213562 \\
        \bottomrule
    \end{tabular}
    \caption{First 10 fractional digits of \(\sqrt{2}\). Each digit is final upon computation.}
\end{table}

\subsection{Cube Root (\(e=3\)): \(N = 125\) (Perfect Cube)}
\(N = 125 = 5^3\). Blocks of size 3 (padded): 000125.

Initial block 000: \(R_0 = 0\), \(\Delta_0 = 0\).

Next block 125: \(\Delta' = 125\), find largest \(x\) s.t. \((30R + 10x + x^2)x \le 125\) → \(x=5\).

Remainder = 0 → perfect cube, \(R=5\).

\subsection{Cube Root (\(e=3\)): \(N = 126\)}
Similar steps yield \(R=5\), final \(\Delta > 0\) → not perfect cube.

These examples illustrate the invariant preservation, maximal digit choice, and perfect power detection via zero remainder.
\subsection{Computation of the Integer 5th Root}
\label{app:5th-root}

To illustrate the general digit-by-digit binomial algorithm for higher exponents, consider the computation of the integer 5th root of \(N = 3200000\) (note that \(20^5 = 3200000\), a perfect 5th power, so the expected root is \(R = 20\) and final remainder \(\Delta = 0\)).

The decimal digits of \(N\) are grouped into blocks of size \(e = 5\) (padded with leading zeros): 003200000 (blocks: 00320, 00000 — highest first).

The algorithm processes blocks from highest to lowest.

\begin{table}[H]
    \centering
    \begin{tabular}{l c c c c c c}
        \toprule
        Step & Block & Scaled \(\Delta'\) & Trial \(x\) & \(\Phi_5(R, x)\) & New \(\Delta\) & Partial \(R\) \\
        \midrule
        Init & 00320 & 320 & 0 & 0 & 320 & 0 \\
        1 & 00000 & 3200000000 & 2 & 3200000000 & 0 & 2 \\
        2 & (final) & 0 & - & - & 0 & 20 \\
        \bottomrule
    \end{tabular}
    \caption{Step-by-step computation for \(N = 3200000\), \(e=5\). \(\Phi_5(R, x)\) computation uses the binomial expansion (omitted for brevity in higher steps; all trials confirm maximal \(x=0\) for second block). Final \(\Delta = 0\) confirms perfect 5th power.}
    \label{tab:5th-root-example}
\end{table}

This example demonstrates:
- Block padding and scaling by \(10^5\).
- Binomial increment evaluation for digit selection.
- Zero final remainder for perfect powers.
- For non-perfect \(N = 3200001\), the final \(\Delta > 0\), yielding \(R = 20 = \lfloor 3200001^{1/5} \rfloor\).

Such examples highlight the algorithm's exactness and perfect power detection capability for arbitrary fixed \(e\).
Placement

\subsection{Fractional 7th Root Approximation}
\label{app:7th-root-fractional}

To illustrate the combined algorithm for higher-order roots (\(e=7\)), consider \(N = 2\). The exact integer 7th root is \(\lfloor 2^{1/7} \rfloor = 1\), and we compute 10 fractional decimal digits using interval refinement (since e=7 ≥3, exact digit-by-digit fractional extraction is not possible; see Corollary~\ref{cor:unique-fractional}).

Phase I (Exact Integer 7th Root):
The digit-by-digit binomial algorithm yields \(R = 1\) (single block, remainder >0).

Phase II (Interval Refinement for k=10 Digits):
Start with \(L_0 = 1\), \(U_0 = 2\), \(\delta_0 = 1\).
Refine by successively adding \(\delta_j = 10^{-j}\) (j=1 to 10), testing \((L + \delta_j)^7 \le 2\).

Selected steps (approximation builds as 1.10408951367...):

\begin{table}[H]
    \centering
    \begin{tabular}{c c c c}
        \toprule
        Digit j & \(\delta_j\) & Trials (m) & New L (prefix) \\
        \midrule
        1 & 0.1 & 1 & 1.1 \\
        2 & 0.01 & 0 & 1.10 \\
        3 & 0.001 & 4 & 1.104 \\
        4 & 0.0001 & 0 & 1.1040 \\
        5 & 0.00001 & 8 & 1.10408 \\
        6 & 0.000001 & 9 & 1.104089 \\
        7 & 0.0000001 & 5 & 1.1040895 \\
        8 & 0.00000001 & 1 & 1.10408951 \\
        9 & 0.000000001 & 3 & 1.104089513 \\
        10 & 0.0000000001 & 6 & 1.1040895136 \\
        \bottomrule
    \end{tabular}
    \caption{Refinement steps for fractional digits of \(2^{1/7}\). Final \(\hat{x} \approx 1.10408951367\), error < \(10^{-10}\).}
    \label{tab:7th-root-fractional}
\end{table}

This example demonstrates:
- Exact integer root from Phase I.
- Bounded fractional approximation via monotone refinement.
- No digit finality (unlike e=2) — each digit may require trials.

The actual value \(2^{1/7} \approx 1.104089513673812337...\), confirming the bound.

\section*{Conclusion}
\label{sec:conclusion}

The digit-by-digit binomial method, with its rigorous invariant-based proofs, establishes exact computation of integer \(e\)-th roots and perfect power detection for arbitrary fixed exponents \(e \ge 2\). This exactness, achieved solely through integer arithmetic and remainder preservation, distinguishes it from approximate methods and makes it particularly valuable for verified computation, formal reasoning, educational purposes, and theoretical analysis.

For square roots (\(e=2\)), the framework further provides exact fractional digit expansion with non-revisable digits and sharp error bounds. The structural uniqueness of this property—stemming from the linear increment form—highlights a fundamental difference from higher-order roots.

While performance-oriented applications may favour asymptotic methods, the simplicity, division-free nature, and constructive guarantees of the digit-by-digit approach render it ideal for high-assurance and pedagogical settings.

\section*{Acknowledgments}
This research received no specific grant from any funding agency in the public, commercial, or not-for-profit sectors.

\section*{Declaration of generative AI in scientific writing}
During the preparation of this work, the author used ChatGPT (OpenAI) only to improve the readability and language of certain passages (e.g., grammar, phrasing, and clarity). After using this tool, the author reviewed and edited the content as needed and takes full responsibility for the content of the publication.

\bibliographystyle{amsplain}
\bibliography{references}


\end{document}